\documentclass[letterpaper, 11pt]{article}

\usepackage{amsthm}
\usepackage{amsmath}
\usepackage{amssymb}
\usepackage{graphicx}
\usepackage{fullpage}
\usepackage{natbib}
\usepackage[english]{babel}
\usepackage{array}
\usepackage{float}
\usepackage{rotating}
\usepackage{mathtools}
\usepackage{cleveref}
\usepackage{algorithm}
\usepackage{algorithmic}
\usepackage{tikz}
\usetikzlibrary{plotmarks}
\usetikzlibrary{patterns}
\usetikzlibrary{shapes}
\usetikzlibrary{shapes.misc}
\usepackage{gnuplot-lua-tikz}

\newtheorem{thm}{Theorem} %[section]

\newtheorem{prop}[thm]{Proposition}

\newtheorem{defn}[thm]{Definition}

\newtheorem{rem}[thm]{Remark}

\usepackage[font=small,skip=3pt]{caption}
\usepackage{subcaption}

\usepackage{xcolor}
\usepackage{setspace}
\usepackage[hyphens]{url}
\usepackage{booktabs}
\usepackage{appendix}
\usepackage{breqn}

\newcolumntype{P}[1]{>{\centering\arraybackslash}p{#1}}
\newcolumntype{M}[1]{>{\centering\arraybackslash}m{#1}}

\renewcommand{\hat}{\widehat}

\usepackage{etoolbox}
\makeatletter
\def\do#1{\@namedef{#1c}{\ensuremath{\mathcal{#1}}}}
\docsvlist{A,B,C,D,E,F,G,H,I,J,K,L,M,N,O,P,Q,R,S,T,U,V,W,X,Y,Z}
\makeatother

\allowdisplaybreaks

\usepackage{authblk}

\title{Best subset selection in linear regression via bi-objective mixed integer linear programming}

\author[$\dag$]{Hadi Charkhgard \thanks{Corresponding author. E-mail address: \textit{hcharkhgard@usf.edu}}}
\author[$\ddag$]{Ali Eshragh}
\affil[$\dag$]{Department of Industrial and Management Systems Engineering, University of South Florida}
\affil[$\ddag$]{School of Mathematical and Physical Sciences, The University of Newcastle, NSW, Australia}
\date{} 

\begin{document}

\doublespacing

\maketitle
\begin{abstract}
We study the problem of choosing the best subset of p features in linear regression given n observations. This problem naturally contains two objective functions including minimizing the amount of bias and minimizing the number of predictors. The existing approaches transform the problem into a single-objective optimization problem. We explain the main weaknesses of existing approaches, and to overcome their drawbacks, we propose a bi-objective mixed integer linear programming approach. A computational study shows the efficacy of the proposed approach. \\

\end{abstract}

%\tableofcontents

%---------------------
\section{Introduction}
%---------------------

Linear regression models should have two important characteristics in practice including \textit{prediction accuracy} and \textit{interpretability} \citep{Tibshirani94}. The traditional approach of constructing regression models is to minimize the sum of squared residuals. It is evident that models obtained in this approach have low biases. However, their prediction accuracy can be low due to their large variances. Furthermore, models constructed by this approach may contain a large number of predictors and so data analysts struggle in interpreting them. 

In general, reducing the number of predictors in a regression model can improve not only the \textit{interpretability} but also, sometimes, the \textit{prediction accuracy} by reducing the variance \citep{Tibshirani94}. Hence, there is often a trade-off between the amount of bias and the practical characteristics of a regression model. In other words, finding a desirable regression model is naturally a bi-objective optimization problem that minimizes the amount of bias and the number of predictors, simultaneously. This classical Statistics problem, that is 
\begin{quote}
	\emph{finding the best subset of $p$ potential predictors and estimating their coefficients in a linear regression model given $n$ observations}
\end{quote}
is called the \textit{Best Subset Selection Problem} (\texttt{BSSP}). Due to the intensive computation of solving bi-objective optimization problems, \texttt{BSSP} has always been transformed into a single-objective optimization problem (e.g., see \cite{Ren2010} and \cite{bertsimas2016}). Such transformation may cause practical issues that will be discussed in this paper. 

Albeit, recent algorithmic and theoretical advances have made bi-objective optimization problems computationally tractable in practice. More precisely, under some mild conditions, we are now able to solve them reasonably fast. To the best of our knowledge, this paper is the first work to tackle \texttt{BSSP} utilizing a bi-objective optimization approach. 

The structure of the paper is organized as follows: In Section \ref{sec:SinglevsBi-objective}, two main single-objective approaches for \texttt{BSSP} are presented and their practical drawbacks are discussed. In Section~\ref{sec:formulation}, a new bi-objective optimization formulation for \texttt{BSSP} is developed. In Section~\ref{sec:Comp}, the computational results are reported. Finally, in Section~\ref{sec:Finalremarks}, some concluding remarks are provided.

%-----------------------------------------------------------------
\section{Single-objective models for \texttt{BSSP}}
\label{sec:SinglevsBi-objective}
%-----------------------------------------------------------------

As discussed in Section 1, \texttt{BSSP} is naturally a bi-objective optimization problem which can be stated as $\min_{\hat{\boldsymbol{\beta}}\in \mathcal{F}}\  \{z_1(\hat{\boldsymbol{\beta}}),z_2(\hat{\boldsymbol{\beta}})\}$
where $\mathcal{F}$ is the feasible set of parameter estimator vector $\hat{\boldsymbol{\beta}}$, $z_1(\hat{\boldsymbol{\beta}})$ is the total bias and $z_2(\hat{\boldsymbol{\beta}})$ is the number of predictors. In the literature, the following two approaches have widely been used to convert \texttt{BSSP} to a single-objective optimization problem:  

\begin{enumerate}
	\item [(i)] \textit{The weighted sum approach:} Given some $\lambda> 0$, \texttt{BSSP} has been reformulated as 
	$$
	\min_{\hat{\boldsymbol{\beta}}\in \mathcal{F}}\   z_1(\hat{\boldsymbol{\beta}})+\lambda z_2(\hat{\boldsymbol{\beta}}).$$
	\item [(ii)] \textit{The goal programming approach:} Given some $k\in \mathbb{Z}_\ge$, \texttt{BSSP} has been reformulated as 
	$$
	\min_{\hat{\boldsymbol{\beta}}\in \mathcal{F}:\  z_2(\hat{\boldsymbol{\beta}})\le k}\   z_1(\hat{\boldsymbol{\beta}}).
	$$
\end{enumerate}
For further details, interested readers are referred to \cite{Chen1998}, \cite{Meinshausen2006}, \cite{zhang2008}, \cite{Bickel2009}, \cite{Candes2009}, \cite{Ren2010}, \cite{Ciuperca2014} and \cite{Wang2017} for the weighted sum approach, and \cite{Miller2002} and \cite{bertsimas2016} for the goal programming approach. Although, those two optimization programs (i) and (ii) can be solved significantly faster than a bi-objective optimization problem, their drawbacks are explained and illustrated here. However, we first need to present some notation and definitions from bi-objective optimization. 

%---------------------------
\subsection{Preliminaries}
\label{sec:Preliminaries}
%---------------------------

A {\em Bi-Objective Mixed Integer Linear Program} (BOMILP) can be stated as follows:
\begin{equation}
\label{eq:MO}
\min_{(\boldsymbol{x}_1,\boldsymbol{x}_2) \in \mathcal{X}} \ \{z_1(\boldsymbol{x}_1,\boldsymbol{x}_2),z_2(\boldsymbol{x}_1,\boldsymbol{x}_2)\},
\end{equation}

\noindent where $\mathcal{X}:=\big\{(\boldsymbol{x}_1,\boldsymbol{x}_2)\in \mathbb{Z}^{n_1}_\ge \times \mathbb{R}^{n_2}_\ge:A_1\boldsymbol{x}_1+A_2\boldsymbol{x}_2\le \boldsymbol{b}\big\}$ represents the \textit{feasible set in the decision space}, $\mathbb{Z}^{n_1}_\ge:=\{\boldsymbol{s}\in \mathbb{Z}^{n_1}:\boldsymbol{s}\ge \boldsymbol{0}\}$,  $\mathbb{R}^{n_2}_\ge:=\{\boldsymbol{s}\in \mathbb{R}^{n_2}:\boldsymbol{s}\ge \boldsymbol{0}\}$, $A_1\in \mathbb{R}^{m\times n_1}$, $A_2\in \mathbb{R}^{m\times n_2}$, and $\boldsymbol{b}\in \mathbb{R}^{m}$.  It is assumed that $\mathcal{X}$ is \textit{bounded} and $z_i(\boldsymbol{x}_1,\boldsymbol{x}_2)=\boldsymbol{c}^\intercal_{i,1}\boldsymbol{x}_1+\boldsymbol{c}^\intercal_{i,2}\boldsymbol{x}_2$  where $\boldsymbol{c}_{i,1}\in \mathbb{R}^{n_1}$ and $\boldsymbol{c}_{i,2}\in \mathbb{R}^{n_2}$ for $i=1,2$ represents a linear objective function. The image $\mathcal{Y}$ of $\mathcal{X}$ under vector-valued function
$\boldsymbol{z}:=(z_1,z_2)^\intercal$ represents the \textit{feasible set in the objective/criterion space}, that is $\mathcal{Y} := \{\boldsymbol{o}\in \mathbb{R}^2: \boldsymbol{o}=\boldsymbol{z}(\boldsymbol{x}_1,\boldsymbol{x}_2)$ for all $(\boldsymbol{x}_1,\boldsymbol{x}_2)\in \mathcal{X} \}$. Note that BOMILP is called \textit{Bi-objective Linear Program} (BOLP) and \textit{Bi-Objective Integer Linear Program} (BOILP) for the special cases of $n_1=0$ and $n_2=0$, respectively. 

\begin{defn}
	\label{def:nondominant}
	\begin{doublespace}
		A feasible solution $(\boldsymbol{x}_1,\boldsymbol{x}_2)\in \mathcal{X}$ is called \textit{efficient}
		or \textit{Pareto optimal}, if there is no other $(\boldsymbol{x}'_1,\boldsymbol{x}'_2)\in \mathcal{X}$
		such that $z_1(\boldsymbol{x}'_1,\boldsymbol{x}'_2)\le z_1(\boldsymbol{x}_1,\boldsymbol{x}_2)$ and  $z_2(\boldsymbol{x}'_1,\boldsymbol{x}'_2)< z_2(\boldsymbol{x}_1,\boldsymbol{x}_2)$ or $z_1(\boldsymbol{x}'_1,\boldsymbol{x}'_2)< z_1(\boldsymbol{x}_1,\boldsymbol{x}_2)$ and  $z_2(\boldsymbol{x}'_1,\boldsymbol{x}'_2)\le z_2(\boldsymbol{x}_1,\boldsymbol{x}_2)$.
		If $(\boldsymbol{x}_1,\boldsymbol{x}_2)$ is efficient, then $\boldsymbol{z}(\boldsymbol{x}_1,\boldsymbol{x}_2)$ is called a \textit{nondominated
			point}. The set of all efficient solutions is
		denoted by $\mathcal{X}_E$. The set of all nondominated points $
		\boldsymbol{z}(\boldsymbol{x}_1,\boldsymbol{x}_2)$ for $(\boldsymbol{x}_1,\boldsymbol{x}_2) \in \mathcal{X}_E$ is denoted by
		$\mathcal{Y}_N$ and referred to as the \textit{nondominated
			frontier}.
	\end{doublespace}
\end{defn}
\vspace{-0.5in}
\begin{defn}
	\label{def:SupportedPoints}
	\begin{doublespace}
		If there exists a vector $(\lambda_1,\lambda_2)^\intercal \in
		\mathbb{R}^2_{>}:=\{\boldsymbol{s}\in \mathbb{R}^{2}:\boldsymbol{s}> \boldsymbol{0}\}$ such that $(\boldsymbol{x}_1^*,\boldsymbol{x}_2^*)\in\arg\min_{(\boldsymbol{x}_1,\boldsymbol{x}_2)\in\mathcal{X}}\lambda_1z_1(\boldsymbol{x}_1,\boldsymbol{x}_2)\allowbreak+\lambda_2z_2(\boldsymbol{x}_1,\boldsymbol{x}_2)$, then $(\boldsymbol{x}_1^*,\boldsymbol{x}_2^*)$ is called a
		\textit{supported efficient solution} and $\boldsymbol{z}(\boldsymbol{x}_1^*,\boldsymbol{x}_2^*)$ is called a
		\textit{supported nondominated point}. 
	\end{doublespace}
\end{defn}
\vspace{-0.5in}
\begin{defn}
	\label{def:ExteremePoints}
	\begin{doublespace}
		Let $\mathcal{Y}^{e}$ be the set of extreme points of the convex hull of $\mathcal{Y}$, that is the smallest convex set containing the set $\mathcal{Y}$. A point $\boldsymbol{z}(\boldsymbol{x}_1,\boldsymbol{x}_2) \in \mathcal{Y}$ is called an
		\textit{extreme supported nondominated point}, if $\boldsymbol{z}(\boldsymbol{x}_1,\boldsymbol{x}_2)$ is a supported nondominated point and $\boldsymbol{z}(\boldsymbol{x}_1,\boldsymbol{x}_2) \in \mathcal{Y}^e$.
	\end{doublespace}
	\vspace{-0.2 in}
\end{defn}

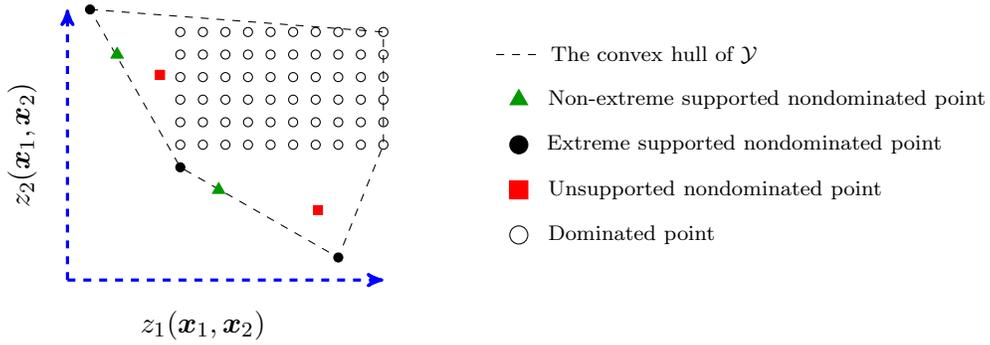
\begin{figure}[!htbp]
	\begin{center}
		\begin{tikzpicture}[scale=0.30,
		axis/.style={very thick, ->, >=stealth',dashed,draw=blue},line/.style={very thick},]			 	
		\tikzset{cross/.style={very thick,cross out, draw=black, minimum size=1.5*(#1-\pgflinewidth), inner sep=0pt, outer sep=0pt},
			%default radius will be 1pt.
			cross/.default={5pt}}
		%----------------------------------------------------------------------------
		%----------------LEGEND------------------
		\draw [dashed] (18,9)--(20,9);
		\node[] at (25,9)  {\scriptsize The convex hull of $\mathcal{Y}$};
		\node[fill=black!40!green,regular polygon, regular polygon sides=3,inner sep=1 .5pt] at (19,7) {};
		\node[] at (30,7)  {\scriptsize Non-extreme supported nondominated point};
		\draw [fill=black] (19,5) circle (0.4);
		\node[] at (29,5)  {\scriptsize Extreme supported nondominated point};
		\draw [red,fill=red] (18.6,2.6) rectangle (19.4,3.4);
		\node[] at (27.7,3)  {\scriptsize Unsupported nondominated point};
		\draw [] (19,1) circle (0.4);
		\node[] at (24,1)  {\scriptsize Dominated point};			
		\draw [dashed] (0,11)--(13,10)--(13,5)--(11,0)--(4,4)--(0,11);
		\draw [axis] (-1,-1)--(-1,11);
		\draw [axis] (-1,-1)--(13,-1);
		\node[rotate=90] at (-3,5)  {$z_2(\boldsymbol{x}_1,\boldsymbol{x}_2)$};
		\node[] at (5,-3)  {$z_1(\boldsymbol{x}_1,\boldsymbol{x}_2)$};
		\draw [fill=black] (0,11) circle (0.2);
		\draw [fill=black] (4,4) circle (0.2);
		\draw [fill=black] (11,0) circle (0.2);
		\draw [red,fill=red] (2.9,7.9) rectangle (3.3,8.3);
		\draw [red,fill=red] (9.9,1.9) rectangle (10.3,2.3);
		\node[fill=black!40!green,regular polygon, regular polygon sides=3,inner sep=1.1pt] at (1.2,9) {};
		\node[fill=black!40!green,regular polygon, regular polygon sides=3,inner sep=1.1pt] at (5.7,3) {};
		\draw [] (4,5) circle (0.2);
		\draw [] (5,5) circle (0.2);
		\draw [] (6,5) circle (0.2);
		\draw [] (7,5) circle (0.2);
		\draw [] (8,5) circle (0.2);
		\draw [] (9,5) circle (0.2);
		\draw [] (10,5) circle (0.2);
		\draw [] (11,5) circle (0.2);
		\draw [] (12,5) circle (0.2);
		\draw [] (13,5) circle (0.2);
		\draw [] (4,6) circle (0.2);
		\draw [] (5,6) circle (0.2);
		\draw [] (6,6) circle (0.2);
		\draw [] (7,6) circle (0.2);
		\draw [] (8,6) circle (0.2);
		\draw [] (9,6) circle (0.2);
		\draw [] (10,6) circle (0.2);
		\draw [] (11,6) circle (0.2);
		\draw [] (12,6) circle (0.2);
		\draw [] (13,6) circle (0.2);
		\draw [] (4,7) circle (0.2);
		\draw [] (5,7) circle (0.2);
		\draw [] (6,7) circle (0.2);
		\draw [] (7,7) circle (0.2);
		\draw [] (8,7) circle (0.2);
		\draw [] (9,7) circle (0.2);
		\draw [] (10,7) circle (0.2);
		\draw [] (11,7) circle (0.2);
		\draw [] (12,7) circle (0.2);
		\draw [] (13,7) circle (0.2);
		\draw [] (4,8) circle (0.2);
		\draw [] (5,8) circle (0.2);
		\draw [] (6,8) circle (0.2);
		\draw [] (7,8) circle (0.2);
		\draw [] (8,8) circle (0.2);
		\draw [] (9,8) circle (0.2);
		\draw [] (10,8) circle (0.2);
		\draw [] (11,8) circle (0.2);
		\draw [] (12,8) circle (0.2);
		\draw [] (13,8) circle (0.2);
		\draw [] (4,9) circle (0.2);
		\draw [] (5,9) circle (0.2);
		\draw [] (6,9) circle (0.2);
		\draw [] (7,9) circle (0.2);
		\draw [] (8,9) circle (0.2);
		\draw [] (9,9) circle (0.2);
		\draw [] (10,9) circle (0.2);
		\draw [] (11,9) circle (0.2);
		\draw [] (12,9) circle (0.2);
		\draw [] (13,9) circle (0.2);
		\draw [] (4,10) circle (0.2);
		\draw [] (5,10) circle (0.2);
		\draw [] (6,10) circle (0.2);
		\draw [] (7,10) circle (0.2);
		\draw [] (8,10) circle (0.2);
		\draw [] (9,10) circle (0.2);
		\draw [] (10,10) circle (0.2);
		\draw [] (11,10) circle (0.2);
		\draw [] (12,10) circle (0.2);
		\draw [] (13,10) circle (0.2);
		\end{tikzpicture}
	\end{center}
	\caption{An illustration of different types of (feasible) points in the criterion space}
	\label{fig:PointTypes}
\end{figure}

In summary, based on Definition~\ref{def:nondominant}, the elements of $\mathcal{Y}$ can be divided into two groups including dominated and nondominated points. Furthermore, based on Definitions~\ref{def:SupportedPoints} and \ref{def:ExteremePoints}, the nondominated points can be divided into unsupported nondominated points, non-extreme supported nondominated points and extreme supported nondominated points. Overall, bi-objective optimization problems are concerned with finding all elements of $\mathcal{Y}_N$, that is all  nondominated points, including supported and unsupported nondominated points. An illustration of the set $\mathcal{Y}$ and its corresponding categories are shown in Figure~\ref{fig:PointTypes}.

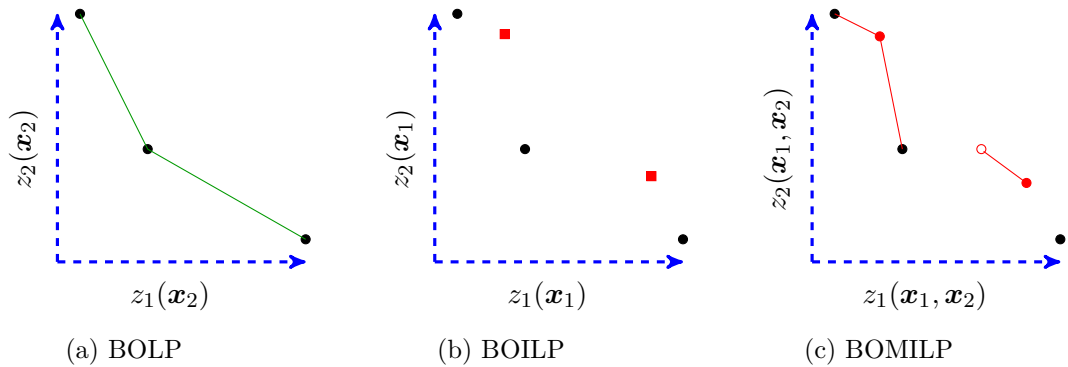
\begin{figure}
	\begin{center}
		\begin{minipage}{0.20\linewidth}
			\begin{tikzpicture}[scale=0.30,
			axis/.style={very thick, ->, >=stealth',dashed,draw=blue},line/.style={very thick},]
			\tikzset{cross/.style={very thick,cross out, draw=black, minimum size=1.5*(#1-\pgflinewidth), inner sep=0pt, outer sep=0pt},
				%default radius will be 1pt.
				cross/.default={5pt}}
			%----------------------------------------------------------------------------
			\draw [axis] (0,0)--(0,11);
			\draw [axis] (0,0)--(11,0);
			\node[] at (5,-1.5)  {$z_1(\boldsymbol{x}_2)$};
			\node[rotate=90] at (-1.5,5)  {$z_2(\boldsymbol{x}_2)$};
			\draw [fill=black] (1,11) circle (0.2);
			\draw [fill=black] (4,5) circle (0.2);
			\draw [fill=black] (11,1) circle (0.2);
			
			\draw [black!40!green] (1,11)--(4,5)--(11,1);
			\end{tikzpicture}
			\subcaption{ BOLP}
			\label{fig:frontierLP}
		\end{minipage}
		\hfil
		\begin{minipage}{0.20\linewidth}
			\begin{tikzpicture}[scale=0.30,
			axis/.style={very thick, ->, >=stealth',dashed,draw=blue},line/.style={very thick},]
			\tikzset{cross/.style={very thick,cross out, draw=black, minimum size=1.5*(#1-\pgflinewidth), inner sep=0pt, outer sep=0pt},
				%default radius will be 1pt.
				cross/.default={5pt}}
			%----------------------------------------------------------------------------
			\draw [axis] (0,0)--(0,11);
			\draw [axis] (0,0)--(11,0);
			\node[] at (5,-1.5)  {$z_1(\boldsymbol{x}_1)$};
			\node[rotate=90] at (-1.5,5)  {$z_2(\boldsymbol{x}_1)$};
			\draw [fill=black] (1,11) circle (0.2);
			\draw [red, fill=red] (2.9,9.9) rectangle (3.3,10.3);
			\draw [fill=black] (4,5) circle (0.2);
			\draw [red, fill=red] (9.4,3.6) rectangle (9.8,4);
			\draw [fill=black] (11,1) circle (0.2);
			\end{tikzpicture}
			\subcaption{BOILP}
			\label{fig:frontierIP}
		\end{minipage}
		\hfil
		\begin{minipage}{0.20\linewidth}
			\begin{tikzpicture}[scale=0.30,
			axis/.style={very thick, ->, >=stealth',dashed,draw=blue},line/.style={very thick},]
			\tikzset{cross/.style={very thick,cross out, draw=black, minimum size=1.5*(#1-\pgflinewidth), inner sep=0pt, outer sep=0pt},
				%default radius will be 1pt.
				cross/.default={5pt}}
			%----------------------------------------------------------------------------
			\draw [axis] (0,0)--(0,11);
			\draw [axis] (0,0)--(11,0);
			\node[] at (5,-1.5)  {$z_1(\boldsymbol{x}_1,\boldsymbol{x}_2)$};
			\node[rotate=90] at (-1.5,5)  {$z_2(\boldsymbol{x}_1,\boldsymbol{x}_2)$};
			
			\draw [fill=black] (1,11) circle (0.2);
			\draw [color=red,fill=red] (3,10) circle (0.2);
			\draw [fill=black] (4,5) circle (0.2);
			\draw [color=red,fill=white] (7.5,5) circle (0.2);
			\draw [color=red,fill=red] (9.5,3.5) circle (0.2);
			\draw [fill=black] (11,1) circle (0.2);

			\draw [color=red] (1,11)-- (3,10)--(4,5);
			\draw [color=red] (7.6,4.9)-- (9.5,3.5);
			%-------------------------------------------
			\end{tikzpicture}
			\subcaption{BOMILP}
			\label{fig:frontierMIP}
		\end{minipage}\\
	\end{center}
	\caption{An illustration of the nondominated frontier} 
	\label{fig:Frontier}
\end{figure}
It is well-known that in a BOLP, both the set of efficient solutions $\mathcal{X}_E$ and the set of nondominated points $\mathcal{Y}_N$ are supported and
connected \citep{SAYIN199687}. Consequently, to
describe all nondominated points in a BOLP, it suffices to find all extreme supported nondominated points. A typical illustration of the nondominated frontier of a BOLP is displayed in Figure~\ref{fig:frontierLP}.

Since we assume that $\mathcal{X}$ is bounded, the set of nondominated points of a BOILP is finite. However, due to the existence of unsupported nondominated points in a BOILP, finding all nondominated points is more challenging than in a BOLP. A typical nondominated frontier of a BOILP is shown in Figure~\ref{fig:frontierIP} where the rectangles are unsupported nondominated points. 

Finding all nondominated points of a BOMILP is even more challenging. Nonetheless, if at most one of the objective functions of a BOMILP contains continuous decision variables, then the set of nondominated points is finite and BOILP solution approaches can be utilized to solve it \citep{SAD2014}. However, in all other cases that more than one objective function contain continuous decision variables, the nondominated frontier of a BOMILP may contain connected parts as well as supported and unsupported nondominated points. Therefore, in these cases, the set of nondominated points may not be finite and BOILP algorithms cannot be applied to solve them anymore. A typical nondominated frontier of a BOMILP is illustrated in Figure~\ref{fig:frontierMIP}, where even half-open (or open) line segments may exist in the nondominated frontier. Interested readers are referred to \cite{Ghosh2014}, \cite{Hamacher2007} and \cite{ Boland2015BBM,Boland2015TSM} for further discussions on the properties of BOILPs and BOMILPs and algorithms to solve them.

\begin{figure}[!htbp]
	\begin{center}
		\begin{minipage}{0.30\linewidth}
			\begin{tikzpicture}[scale=0.30,
			axis/.style={very thick, ->, >=stealth',dashed,draw=blue},line/.style={very thick},]
			\tikzset{cross/.style={very thick,cross out, draw=black, minimum size=1.5*(#1-\pgflinewidth), inner sep=0pt, outer sep=0pt},
				%default radius will be 1pt.
				cross/.default={5pt}}
			%----------------------------------------------------------------------------
			\draw [axis] (-1,-1)--(-1,11);
			\draw [axis] (-1,-1)--(11,-1);
			\node[] at (-3,5)  {$z_2(\hat{\boldsymbol{\beta})}$};
			\node[] at (5,-3)  {$z_1(\hat{\boldsymbol{\beta}})$};
			\draw [fill=black] (0,11) circle (0.2);
			\draw [red, fill=red] (2.9,9.9) rectangle (3.3,10.3);
			\draw [red, fill=red] (3.9,8.9) rectangle (4.3,9.3);
			\draw [red, fill=red] (4.9,7.9) rectangle (5.3,8.3);
			\draw [red, fill=red] (5.9,6.9) rectangle (6.3,7.3);
			\draw [red, fill=red] (6.9,5.9) rectangle (7.3,6.3);
			\draw [red, fill=red] (7.9,4.9) rectangle (8.3,5.3);
			\draw [red, fill=red] (8.9,3.9) rectangle (9.3,4.3);
			\draw [red, fill=red] (9.9,2.9) rectangle (10.3,3.3);
			\draw [fill=black] (11,0) circle (0.2);
			\draw [] (10,5) circle (0.2);
			\draw [] (11,5) circle (0.2);
			\draw [] (12,5) circle (0.2);
			\draw [] (13,5) circle (0.2);
			\draw [] (10,6) circle (0.2);
			\draw [] (11,6) circle (0.2);
			\draw [] (12,6) circle (0.2);
			\draw [] (13,6) circle (0.2);
			\draw [] (10,7) circle (0.2);
			\draw [] (11,7) circle (0.2);
			\draw [] (12,7) circle (0.2);
			\draw [] (13,7) circle (0.2);
			\draw [] (10,8) circle (0.2);
			\draw [] (11,8) circle (0.2);
			\draw [] (12,8) circle (0.2);
			\draw [] (13,8) circle (0.2);
			\draw [] (10,9) circle (0.2);
			\draw [] (11,9) circle (0.2);
			\draw [] (12,9) circle (0.2);
			\draw [] (13,9) circle (0.2);
			\draw [] (10,10) circle (0.2);
			\draw [] (11,10) circle (0.2);
			\draw [] (12,10) circle (0.2);
			\draw [] (13,10) circle (0.2);
			\draw [dashed] (0,11)--(13,10)--(13,5)--(11,0)--(0,11);
			\end{tikzpicture}
			\subcaption{ The wighted sum approach fails to find rectangles}
			\label{fig:WSA-Weakness}
		\end{minipage}
		\hfil
		\begin{minipage}{0.30\linewidth}
			\begin{tikzpicture}[scale=0.30,
			axis/.style={very thick, ->, >=stealth',dashed,draw=blue},line/.style={very thick},
			Blackaxis/.style={very thick, ->, >=stealth',draw=black},line/.style={very thick},
			]
			\tikzset{cross/.style={very thick,cross out, draw=black, minimum size=1.5*(#1-\pgflinewidth), inner sep=0pt, outer sep=0pt},
				%default radius will be 1pt.
				cross/.default={5pt}}
			%----------------------------------------------------------------------------
			\draw [axis] (-1,-1)--(-1,11);
			\draw [axis] (-1,-1)--(11,-1);
			\node[] at (-3,5)  {$z_2(\hat{\boldsymbol{\beta}})$};
			\node[] at (5,-3)  {$z_1(\hat{\boldsymbol{\beta}})$};
			\draw [fill=black] (2,2) circle (0.2);
			\draw [] (3,2) circle (0.2);
			\draw [] (4,2) circle (0.2);
			\draw [] (5,2) circle (0.2);
			\draw [] (6,2) circle (0.2);
			\draw [] (7,2) circle (0.2);
			\draw [] (8,2) circle (0.2);
			\draw [] (1.8,2.8) rectangle (2.2,3.2);
			\draw [] (3,3) circle (0.2);
			\draw [] (4,3) circle (0.2);
			\draw [] (5,3) circle (0.2);
			\draw [] (6,3) circle (0.2);
			\draw [] (7,3) circle (0.2);
			\draw [] (8,3) circle (0.2);
			\draw [] (1.8,3.8) rectangle (2.2,4.2);
			\draw [] (3,4) circle (0.2);
			\draw [] (4,4) circle (0.2);
			\draw [] (5,4) circle (0.2);
			\draw [] (6,4) circle (0.2);
			\draw [] (7,4) circle (0.2);
			\draw [] (8,4) circle (0.2);
			\draw [] (1.8,4.8) rectangle (2.2,5.2);
			\draw [] (3,5) circle (0.2);
			\draw [] (4,5) circle (0.2);
			\draw [] (5,5) circle (0.2);
			\draw [] (6,5) circle (0.2);
			\draw [] (7,5) circle (0.2);
			\draw [] (8,5) circle (0.2);
			\draw [] (2,6) circle (0.2);
			\draw [] (3,6) circle (0.2);
			\draw [] (4,6) circle (0.2);
			\draw [] (5,6) circle (0.2);
			\draw [] (6,6) circle (0.2);
			\draw [] (7,6) circle (0.2);
			\draw [] (8,6) circle (0.2);
			\draw [] (2,7) circle (0.2);
			\draw [] (3,6) circle (0.2);
			\draw [] (4,6) circle (0.2);
			\draw [] (5,6) circle (0.2);
			\draw [] (6,6) circle (0.2);
			\draw [] (7,6) circle (0.2);
			\draw [] (8,6) circle (0.2);
			\draw [] (3,7) circle (0.2);
			\draw [] (4,7) circle (0.2);
			\draw [] (5,7) circle (0.2);
			\draw [] (6,7) circle (0.2);
			\draw [] (7,7) circle (0.2);
			\draw [] (8,7) circle (0.2);
			\draw [fill=black] (-1,5.4) rectangle (11,5.6);
			\draw [Blackaxis] (0,5.4)--(0,4.4);
			\draw [Blackaxis] (10,5.4)--(10,4.4);
			\node[] at (13.5 ,5.5)  {\scriptsize $z_2(\hat{\boldsymbol{\beta}})\le k$};
			%\draw [dashed] (2,2)--(2,7)--(8,7)--(8,2)--(2,2);
			\end{tikzpicture}
			\subcaption{ The goal programming approach may incorrectly find rectangles}
			\label{fig:GBA-Weakness}
		\end{minipage}
		\caption{The set of feasible points in the criterion space}
		\label{fig:feasible points}
	\end{center}
\end{figure}
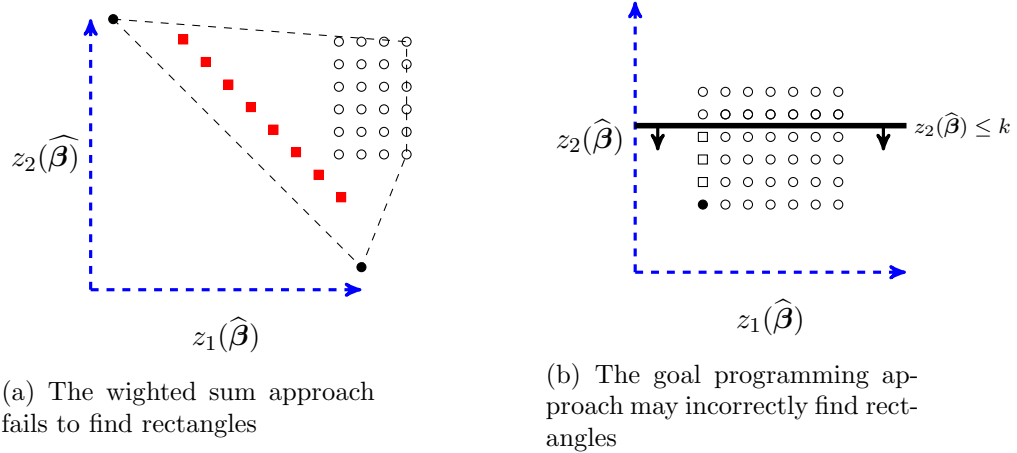

%---------------------------
\subsection{Drawbacks of transforming \texttt{BSSP} into a single-objective model}
\label{sec:Issues}
%---------------------------

Suppose that for each $\hat{\boldsymbol{\beta}}\in \mathcal{F}$, the corresponding point $(z_1(\hat{\boldsymbol{\beta}}),z_2(\hat{\boldsymbol{\beta}}))^\intercal$ is plotted into the criterion space. Figures~\ref{fig:WSA-Weakness} and \ref{fig:GBA-Weakness} show two typical plots of such pairs for all $\hat{\boldsymbol{\beta}}\in \mathcal{F}$. In these two figures, all filled circles and rectangles are nondominated points of the problem and unfilled rectangles and circles are dominated points. In Figure~\ref{fig:WSA-Weakness}, the region defined by the dashed lines is the convex hull of all feasible points. In this case, it is impossible that the wighted sum approach finds the filled rectangles for any arbitrary weight as all filled rectangles are unsupported nondominated points (i.e., they are interior points of the convex hull). So, this illustrates that there may exist many nondominated points, but the weighted sum approach can fail to find most of them for any arbitrary weight. Figure~\ref{fig:GBA-Weakness} is helpful for understanding the main drawback of the goal programming approach. It is obvious that depending on the value of $k$, the goal programming approach may find one of the unfilled rectangles which are dominated points. So, the main drawback of the goal programming approach is that it may even fail to find a nondominated point.     

The main contribution of our research presented here is to overcome both of these disadvantages by utilizing bi-objective optimization techniques. We note that in the literature of \texttt{BSSP}, $z_1(\hat{\boldsymbol{\beta}})$ is mainly defined as the sum of squared residuals. The reason lies in the fact that the sum of squared residuals is a smooth (convex) function. However, to be able to exploit existing bi-objective mixed integer linear programming solvers, we use the sum of absolute residuals for $z_1(\hat{\boldsymbol{\beta}})$. We conclude this section by providing two remarks.
\begin{rem}
	\begin{doublespace}
		If we incorporate additional linear constraints on the vector of parameter estimators of the regression model, $\hat{\boldsymbol{\beta}}$, it will be more likely that the goal programming approach fails to find a nondominated point.   
	\end{doublespace} 
\end{rem}
\vspace{-0.35in}
\begin{rem}
	\begin{doublespace}
		Unlike the weighted sum and goal programming approaches that new parameters $\lambda$ and $k$, respectively, should be employed and tuned by the user, the bi-objective optimization approach does not need any extra parameter. 
	\end{doublespace}
\end{rem}       
\vspace{-0.5in}

%--------------------------------------------------------------------
\section{A bi-objective formulation for \texttt{BSSP}}
\label{sec:formulation}
%--------------------------------------------------------------------

Let $X=[\boldsymbol{x}_1,\dots,\boldsymbol{x}_p]\in \mathbb{R}^{n\times p}$ be the model matrix (it is assumed that $\boldsymbol{x}_1=\boldsymbol{1}$), $\boldsymbol{\beta}\in \mathbb{R}^{p\times 1}$ be the vector of regression coefficients, and $\boldsymbol{y}\in \mathbb{R}^{n\times 1}$ be the response vector. It is assumed that $\boldsymbol{\beta}$ is unknown and should be estimated. Let  $\hat{\boldsymbol{\beta}}\in \mathbb{R}^{p\times 1}$ denote an estimate for $\boldsymbol{\beta}$. To solve \texttt{BSSP} for this set of data, we construct the following BOMILP and denote it by $\texttt{BSSP-BOMILP}$:
\begin{align}
\min \ &\{\sum_{i=1}^n \gamma_i,\  \sum_{j=1}^p r_j\} \label{OF:1}\allowdisplaybreaks\\
\mbox{such that: }  &r_jl_j\le \hat{\beta}_j\le r_ju_j &&\mbox{ for } j=1,\dots,p \label{const:z_j}\allowdisplaybreaks\\
& y_i-\sum_{j=1}^px_{ij}\hat{\beta}_j\le \gamma_i  &&\mbox{ for } i=1,\dots,n \label{const:Abs1}\allowdisplaybreaks\\
& \sum_{j=1}^px_{ij}\hat{\beta}_j-y_i\le \gamma_i  &&\mbox{ for } i=1,\dots,n\label{const:Abs2} \allowdisplaybreaks\\
&r_j\in\{0,1\} &&\mbox{ for } j=1,\dots,p\allowdisplaybreaks\\
& \gamma_i\ge 0  &&\mbox{ for } i=1,\dots,n\\
& \hat{\beta}_j\in \mathbb{R}  &&\mbox{ for } j=1,\dots,p,\label{const:Last}
\end{align}
where $l_j\in \mathbb{R}$ and $u_j\in \mathbb{R}$ are, respectively, a lower bound and an upper bound (known) for $\hat{\beta}_j$ for $j=1,\dots,p$, $\gamma_i$ for $i=1,\dots,n$ is a non-negative continuous variable that captures the value of $|y_i-\sum_{j=1}^px_{ij}\hat{\beta}_j|$ in any efficient solution, and $r_j$ for $j=1,\dots,p$ is a binary decision variable that takes the value of one if $\hat{\beta}_j\ne 0$, implying that the predictor $j$ is active. By these definitions, for any efficient solution, the first objective function, $z_1(\hat{\boldsymbol{\beta}})=\sum_{i=1}^n \gamma_i$, takes the value of the sum of absolute residuals and the second objective function, $z_2(\hat{\boldsymbol{\beta}})=\sum_{j=1}^p r_j$, computes the number of predictors. Constraint \eqref{const:z_j} ensures that if $\hat{\beta}_j\ne 0$ then $r_j=1$ for $j=1,\dots,p$. Constraints \eqref{const:Abs1} and \eqref{const:Abs2} guarantee that $|y_i-\sum_{j=1}^px_{ij}\hat{\beta}_j|\le \gamma_i$ for $i=1,\dots,n$. Note that since we minimize the first objective function, we have  $|y_i-\sum_{j=1}^px_{ij}\hat{\beta}_j|=\gamma_i$ for $i=1,\dots,n$ in an efficient solution.
\begin{rem}
	\begin{doublespace}
		The $\texttt{BSSP-BOMILP}$ can handle additional linear constraints and variables. Furthermore, by choosing tight bounds in Constraint \eqref{const:z_j}, we can speed up the solution time of $\texttt{BSSP-BOMILP}$. Hence, we should try to choose $l_j$/$u_j$ as large/small as possible.
	\end{doublespace}
\end{rem}
\vspace{-0.5in}
\begin{rem}
	\begin{doublespace}
		Since only one of the objective functions in $\texttt{BSSP-BOMILP}$ contains continuous variables, based on our discussion in Section~\ref{sec:Preliminaries}, the set of nondominated points of $\texttt{BSSP-BOMILP}$ is finite. More precisely, the nondominated frontier of $\texttt{BSSP-BOMILP}$ can have at most $p+1$ number of nondominated points as $\sum_{j=1}^p r_j\in\{0,1,\dots,p\}$. So we can use BOILP solvers such as the $\epsilon$-constraint method or the balanced box method to solve $\texttt{BSSP-BOMILP}$ \citep{Chankong1983,Boland2015BBM}.
	\end{doublespace}
\end{rem}
\vspace{-0.5in}
\begin{rem}
	\label{rem:Remark2}
	\begin{doublespace}
		The solution $(\boldsymbol{\gamma}^B,\boldsymbol{r}^B, \hat{\boldsymbol{\beta}}^B):=(|\boldsymbol{y}|,\boldsymbol{0},\boldsymbol{0})$ is a trivial efficient solution of $\texttt{BSSP-BOMILP}$ which attains the minimum possible value for the second objective function. Accordingly, the point $(\sum_{i=1}^n\gamma_i^B,\sum_{j=1}^pr^B_j)=(\sum_{i=1}^n|y_i|,0)$ is a trivial nondominated point of $\texttt{BSSP-BOMILP}$ where there is no parameter selected in the estimated regression model. Hence, we exclude this trivial nondominated point by adding the constraint $\sum_{j=1}^pr_j\ge 1$ to $\texttt{BSSP-BOMILP}$.  
		
	\end{doublespace}
\end{rem} 
\vspace{-0.5in}
\subsection{Bounds for the regression coefficients}

In this section, we develop a data-driven approach to find bounds $l_j$ and $u_j$ for $j=1,\dots,p$ such that $l_j\le \hat{\beta}_j\le u_j$, in the lack of any additional information. To achieve this, we firstly present the following proposition.
\begin{prop}
	\label{prop:bound}
	\begin{doublespace}	
		Let $m$ be the median of response observations $y_{1},\dots,y_{n}$. If $(\boldsymbol{\gamma}^*,\boldsymbol{r}^*, \hat{\boldsymbol{\beta}}^*)$ is an efficient solution of $\texttt{BSSP-BOMILP}$, then $\sum_{i=1}^n \gamma^*_i\le \sum_{i=1}^n|y_i-m|$.
	\end{doublespace}
\end{prop}
\begin{proof}
	\vspace{-0.3in}
	\begin{doublespace}
		Let consider the feasible solution  $(\boldsymbol{\gamma},\boldsymbol{r}, \hat{\boldsymbol{\beta}})$ where $r_1=1$, $\beta_1=m$, $r_j=\beta_j=0$ for $j=2,\dots, p$,  $\gamma_i=|y_i-\sum_{j=1}^px_{ij}\hat{\beta}_j|$ for $i=1,\dots,n$. So, we have $\gamma_i=|y_i-m|$ for $i=1,\dots,n$ because $x_{i1}=1$ for $i=1,\dots,n$ in $\texttt{BSSP-BOMILP}$. Since by Remark~\ref{rem:Remark2}, $\sum_{j=1}^pr^*_j\ge 1=\sum_{j=1}^pr_j$, we must have $\sum_{i=1}^n \gamma^*_i\le \sum_{i=1}^n|y_i-m|=\sum_{i=1}^n \gamma_i$ to keep $(\boldsymbol{\gamma}^*,\boldsymbol{r}^*, \hat{\boldsymbol{\beta}}^*)$ an efficient solution. \qed
	\end{doublespace}
\end{proof}
\vspace{-0.5in}
\begin{rem}
	\label{rem:bound}
	\begin{doublespace}	
		It is readily seen that if we replace $m$ with any other real number, the inequality given in Proposition~\ref{prop:bound} still holds. However, as the minimum of $\sum_{i=1}^n|y_i-\hat{\beta}_1|$ is achieved at $\hat{\beta}_1=m$ \citep{Neil1990}, Proposition~\ref{prop:bound} provides the best upper bound for $\sum_{i=1}^n \gamma^*_i$. 
	\end{doublespace}
\end{rem}
\vspace{-0.4in}

Motivated from Proposition~\ref{prop:bound}, we can solve the following optimization problem to find $u_j$ for $j=1,\dots,p$:
\vspace{-0.1in}
\begin{align}
\label{eq:labelbound}
u_j:=\max\{\hat{\beta}_j: \sum_{i=1}^n|y_i-\sum_{j'=1}^px_{ij'}\hat{\beta}_{j'}|\le \sum_{i=1}^n|y_i-m|,\ \hat{\boldsymbol{\beta}}\in\mathbb{R}^p \}.
\end{align}
There are several ways to transform \eqref{eq:labelbound} into a linear program (e.g., see \cite{Dielman2005}). Here, we propose the following linear programing model:

\begin{equation}
\begin{aligned}
\label{eq:labelboundLP}
u_j:=\max\{\hat{\beta}_j: & \sum_{i=1}^n\gamma_i\le \sum_{i=1}^n|y_i-m|,\allowdisplaybreaks\\ &y_i-\sum_{j'=1}^px_{ij'}\hat{\beta}_{j'}\le \gamma_i \quad \mbox{ for } i=1,\dots,n,\allowdisplaybreaks\\
&\sum_{j'=1}^px_{ij'}\hat{\beta}_{j'}-y_i\le \gamma_i \quad \mbox{ for } i=1,\dots,n,\allowdisplaybreaks\\
& \hat{\boldsymbol{\beta}}\in\mathbb{R}^p,\ \boldsymbol{\gamma}\in \mathbb{R}^n_\ge\}.
\end{aligned}
\end{equation}
It should be noted that \eqref{eq:labelboundLP} is a relaxation of \eqref{eq:labelbound} since $\gamma_i$ over-calculates $|y_i-\sum_{j'=1}^px_{ij'}\hat{\beta}_{j'}|$ for $i=1,\dots,n$. Analogously, $l_j$ for $j=1,\dots,p$ can be computed by changing `$\max$' into `$\min$' in \eqref{eq:labelboundLP}.

%-------------------------------
\section{Computational results}
\label{sec:Comp}
%-------------------------------

We conduct a computational study to show the performance of the $\epsilon$-constraint method on $\texttt{BSSP-BOMILP}$, numerically. We use C++ to code the $\epsilon$-constraint method. In this computational study, the algorithm uses CPLEX 12.7 as the single-objective integer programming solver. All computational experiments are carried out on a Dell PowerEdge R630 with two Intel Xeon E5-2650 2.2 GHz 12-Core Processors (30MB), 128GB RAM, and the RedHat Enterprise Linux 6.8 operating system. We allow CPLEX to employ at most 10 threads at the same time. 

We design six classes of instances, each denoted by $C(p,n)$ where $p\in\{20,40\}$ and $n\in \{2p,3p,4p\}$. Based on this construction, we generate three regression models for each class as follows:
\begin{itemize}
	\item Set all $x_{i1}=1$ and all other $x_{ij}$ for $j>1$ are randomly drawn from the discrete uniform distribution on interval $[-50,50]$; 
	\item Two steps are taken to construct $y_i$ for $i=1,\dots,n$: (1) A vector $\boldsymbol{\beta}$ is generated such that the two third of its components are zero, and the others are randomly drawn from the uniform distribution on interval $(0,1)$; (2) Set $y_i=\sum_{j=1}^px_{ij}\beta_j+\varepsilon_i$ (with at most one decimal place) where $\varepsilon_i$ is randomly generated from the standard normal distribution;	
	\item Optimal values of $l_j$ and $u_j$ for $j=1,\dots,p$ are computed by solving \eqref{eq:labelboundLP}.
\end{itemize}

\begin{table}[!htbp]
	\centering
	\scriptsize
	\caption{Numerical results obtained by running the $\epsilon$-constraint method}
	\begin{tabular}{|c|r|r|r|r|r|r|c|c|}
		\hline
		\multicolumn{ 1}{|c|}{\textbf{Class}} & \multicolumn{ 2}{c|}{\textbf{Instance 1}} & \multicolumn{ 2}{c|}{\textbf{Instance 2}} & \multicolumn{ 2}{c|}{\textbf{Instance 3}} & \multicolumn{ 2}{c|}{\textbf{Average}} \\ \cline{ 2- 9}
		\multicolumn{ 1}{|c|}{} & \multicolumn{1}{c|}{\textbf{Time(Sec.)}} & \multicolumn{1}{c|}{\textbf{\#NDPs}} & \multicolumn{1}{c|}{\textbf{Time(Sec.)}} & \multicolumn{1}{c|}{\textbf{\#NDPs}} & \multicolumn{1}{c|}{\textbf{Time(Sec.)}} & \multicolumn{1}{c|}{\textbf{\#NDPs}} & \textbf{Time(Sec.)} & \textbf{\#NDPs} \\ \hline
		\textbf{C(20,40)} & 4.1 & 21 & 3.7 & 21 & 3.8 & 21 & \textbf{3.8} & \textbf{21.0} \\ \hline
		\textbf{C(20,60)} & 4.6 & 21 & 5.4 & 21 & 4.3 & 21 & \textbf{4.8} & \textbf{21.0} \\ \hline
		\textbf{C(20,80)} & 5.2 & 21 & 6.0 & 21 & 6.1 & 21 & \textbf{5.8} & \textbf{21.0} \\ \hline
		\textbf{C(40,80)} & 264.4 & 41 & 385.5 & 41 & 290.6 & 41 & \textbf{313.5} & \textbf{41.0} \\ \hline
		\textbf{C(40,120)} & 313.0 & 41 & 921.5 & 41 & 247.0 & 41 & \textbf{493.8} & \textbf{41.0} \\ \hline
		\textbf{C(40,160)} & 275.6 & 41 & 327.9 & 41 & 591.0 & 41 & \textbf{398.2} & \textbf{41.0} \\ \hline
	\end{tabular}
	\label{Tab:Results}
\end{table}
Table~\ref{Tab:Results} reports the numerical results for all 18 instances. For each instance, there are two columns `Time(Sec.)' and `\#NDPs' showing the solution time in seconds and the number of nondominated points, respectively. All nondominated points can be found for instances with $p=20$ and $p=40$ in about 5 seconds and 7 minutes in average, respectively. 

\begin{figure}[!htbp]
	\begin{center}
		\begin{tikzpicture}[gnuplot,scale=0.7, axis/.style={very thick, ->, >=stealth',dashed,draw=blue},line/.style={very thick}]
		%% generated with GNUPLOT 4.6p4 (Lua 5.1; terminal rev. 99, script rev. 100)
		%% Sat 13 May 2017 11:29:01 AM EDT
		\path (0.000,0.000) rectangle (12.500,8.750);
		\gpcolor{color=gp lt color border}
		\gpsetlinetype{gp lt border}
		\gpsetlinewidth{1.00}
		\draw[gp path] (1.320,0.985)--(1.500,0.985);
		\draw[gp path] (11.947,0.985)--(11.767,0.985);
		\node[gp node right] at (1.136,0.985) {\scriptsize 0};
		\draw[gp path] (1.320,2.834)--(1.500,2.834);
		\draw[gp path] (11.947,2.834)--(11.767,2.834);
		\node[gp node right] at (1.136,2.834) {\scriptsize 5};
		\draw[gp path] (1.320,4.683)--(1.500,4.683);
		\draw[gp path] (11.947,4.683)--(11.767,4.683);
		\node[gp node right] at (1.136,4.683) {\scriptsize 10};
		\draw[gp path] (1.320,6.532)--(1.500,6.532);
		\draw[gp path] (11.947,6.532)--(11.767,6.532);
		\node[gp node right] at (1.136,6.532) {\scriptsize 15};
		\draw[gp path] (1.320,8.381)--(1.500,8.381);
		\draw[gp path] (11.947,8.381)--(11.767,8.381);
		\node[gp node right] at (1.136,8.381) { \scriptsize 20};
		\draw[gp path] (1.320,0.985)--(1.320,1.165);
		\draw[gp path] (1.320,8.381)--(1.320,8.201);
		\node[gp node center] at (1.320,0.677) { \scriptsize 0};
		\draw[gp path] (2.383,0.985)--(2.383,1.165);
		\draw[gp path] (2.383,8.381)--(2.383,8.201);
		\node[gp node center] at (2.383,0.677) { \scriptsize 100};
		\draw[gp path] (3.445,0.985)--(3.445,1.165);
		\draw[gp path] (3.445,8.381)--(3.445,8.201);
		\node[gp node center] at (3.445,0.677) { \scriptsize 200};
		\draw[gp path] (4.508,0.985)--(4.508,1.165);
		\draw[gp path] (4.508,8.381)--(4.508,8.201);
		\node[gp node center] at (4.508,0.677) { \scriptsize 300};
		\draw[gp path] (5.571,0.985)--(5.571,1.165);
		\draw[gp path] (5.571,8.381)--(5.571,8.201);
		\node[gp node center] at (5.571,0.677) { \scriptsize 400};
		\draw[gp path] (6.634,0.985)--(6.634,1.165);
		\draw[gp path] (6.634,8.381)--(6.634,8.201);
		\node[gp node center] at (6.634,0.677) { \scriptsize 500};
		\draw[gp path] (7.696,0.985)--(7.696,1.165);
		\draw[gp path] (7.696,8.381)--(7.696,8.201);
		\node[gp node center] at (7.696,0.677) { \scriptsize 600};
		\draw[gp path] (8.759,0.985)--(8.759,1.165);
		\draw[gp path] (8.759,8.381)--(8.759,8.201);
		\node[gp node center] at (8.759,0.677) { \scriptsize 700};
		\draw[gp path] (9.822,0.985)--(9.822,1.165);
		\draw[gp path] (9.822,8.381)--(9.822,8.201);
		\node[gp node center] at (9.822,0.677) { \scriptsize 800};
		\draw[gp path] (10.884,0.985)--(10.884,1.165);
		\draw[gp path] (10.884,8.381)--(10.884,8.201);
		\node[gp node center] at (10.884,0.677) { \scriptsize 900};
		\draw[gp path] (11.947,0.985)--(11.947,1.165);
		\draw[gp path] (11.947,8.381)--(11.947,8.201);
		\node[gp node center] at (11.947,0.677) {\scriptsize 1,000};
		\draw[gp path] (1.320,8.381)--(1.320,0.985)--(11.947,0.985)--(11.947,8.381)--cycle;
		\node[gp node center,rotate=-270] at (0.0,4.683) {\scriptsize The number of predictors $z_2(\hat{\boldsymbol{\beta}})$};
		\node[gp node center] at (6.633,0.0) {\scriptsize The total bias  $z_1(\hat{\boldsymbol{\beta}})$};
		\draw [axis] (3.3,7.2)--(1.8,8.2);
		\node[] at (3.6,7.1)  {\scriptsize $\boldsymbol{T}$};
		\draw [axis] (10.790,2)--(11.69,1.1);
		\node[] at (10.5,2.2)  {\scriptsize $\boldsymbol{B}$};
		\draw [axis] (3.3,5.2)--(1.8,4);
		\node[] at (6.8,6)  {\scriptsize Linear regression model of the desirable point:};
		\node[] at (6.8,5.5)  {\scriptsize $y=0.4+0.87x_2+0.65x_3+0.68x_4+$};
		\node[] at (6.8,5.0)  {\scriptsize $0.43x_5+ 0.51x_6+0.24x_7+0.01x_{14}$};
		\node[] at (4,2)  {\scriptsize (imaginary) ideal point};
		\draw [axis] (3,1.8)--(1.7326,1.1);
		\gpsetpointsize{6.00}
		\gppoint{gp mark 2}{(1.5326,0.97)}
		\gpsetpointsize{6.00}
		\gppoint{gp mark 7}{(1.561,8.381)}
		\gppoint{gp mark 7}{(1.561,8.011)}
		\gppoint{gp mark 7}{(1.561,7.641)}
		\gppoint{gp mark 7}{(1.562,7.272)}
		\gppoint{gp mark 7}{(1.566,6.532)}
		\gppoint{gp mark 7}{(1.573,5.792)}
		\gppoint{gp mark 7}{(1.583,5.053)}
		\gppoint{gp mark 7}{(1.609,3.943)}
		\gppoint{gp mark 7}{(1.637,3.574)}
		\gppoint{gp mark 7}{(1.683,3.204)}
		\gppoint{gp mark 7}{(8.478,1.355)}
		\gppoint{gp mark 7}{(11.790,0.985)}
		\gpcolor{color=gp lt color 0}
		\gpsetpointsize{4.00}
		\gppoint{gp mark 5}{(1.564,6.902)}
		\gppoint{gp mark 5}{(1.569,6.162)}
		\gppoint{gp mark 5}{(1.580,5.423)}
		\gppoint{gp mark 5}{(1.593,4.683)}
		\gppoint{gp mark 5}{(1.604,4.313)}
		\gppoint{gp mark 5}{(3.611,2.834)}
		\gppoint{gp mark 5}{(6.009,2.464)}
		\gppoint{gp mark 5}{(6.950,2.094)}
		\gppoint{gp mark 5}{(7.701,1.725)}
		\gpcolor{color=gp lt color border}
		\draw[gp path] (1.320,8.381)--(1.320,0.985)--(11.947,0.985)--(11.947,8.381)--cycle;
		%% coordinates of the plot area
		\gpdefrectangularnode{gp plot 1}{\pgfpoint{1.320cm}{0.985cm}}{\pgfpoint{11.947cm}{8.381cm}}
		\end{tikzpicture}
		\caption{The nondominated frontier of the Instance 1 from the Class $C(20,40)$}
		\label{fig:NDF-Example}
	\end{center}
\end{figure}
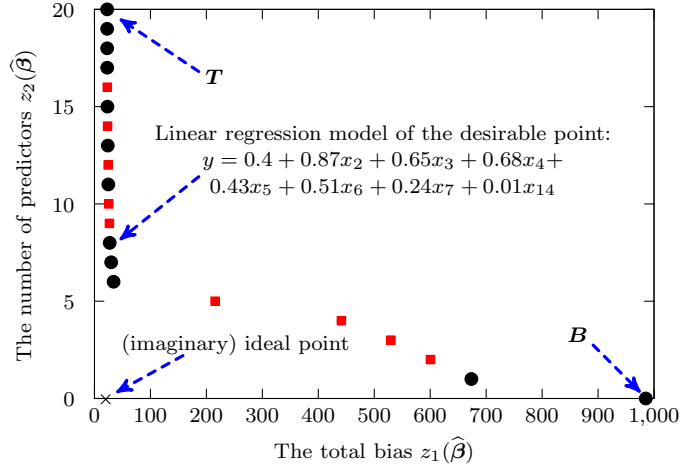
To highlight the drawbacks of existing approaches including the weighted sum approach and the goal programming approach, the nondominated frontier of the Instance 1 from the Class $C(20,40)$ is illustrated in Figure~\ref{fig:NDF-Example}. The filled rectangles and circles are unsupported and supported nondominated points, respectively. As we discussed previously, it is impossible to find any of the unsupported nondominated points using the weighted sum approach. Also, observe that many of the nondominated points lies on an almost vertical line. This implies that all these points are almost optimal for the goal programming approach when $k=7,\dots,20$.

We note that selecting a desirable nondominated point in the nondominated frontier depends on decision maker(s). Here, we introduce a heuristic algorithm to do so. Let $\boldsymbol{T}=(T_1,T_2)^\intercal\in\mathbb{R}^2$ and $\boldsymbol{B}=(B_1,B_2)^\intercal\in\mathbb{R}^2$ be the top and bottom endpoints of the nondominated frontier. One may simply choose the point that has the minimum Euclidean distance from the (imaginary) \textit{ideal} point, that is $(T_1,B_2)^\intercal$. Based on this algorithm, in Figure~\ref{fig:NDF-Example}, the (imaginary) ideal point is $(22.6,0)^\intercal$ and the closest nondominated point to it is $(27.2, 8
)^\intercal$. The generated instance that we discuss in Figure~\ref{fig:NDF-Example} is  $y=0.42+ 0.86x_2+0.63x_3+ 0.68x_4+0.42x_5+0.50x_6+0.25x_7$ and the estimated linear regression model corresponding to the selected nondominated point is $y=0.4+0.87x_2+0.65x_3+0.68x_4+0.43x_5+ 0.51x_6+0.24x_7+0.01x_{14}$, which are very close together.

%--------------------
\section{Conclusion}
\label{sec:Finalremarks}
%--------------------

In this paper, we discussed the practical drawbacks of transforming the best subset selection problem in linear regression modeling into a single-objective optimization problem. To resolve those disadvantages, we developed a new bi-objective optimization model for \texttt{BSSP}. The efficacy of this new model was shown through numerical results. Hopefully, the simplicity, versatility and performance of this new bi-objective optimization model encourage practitioners to consider it for constructing linear regression models.

\bibliographystyle{elsarticle-harv} 
\bibliography{references.bib}  

\begin{thebibliography}{20}
\expandafter\ifx\csname natexlab\endcsname\relax\def\natexlab#1{#1}\fi
\expandafter\ifx\csname url\endcsname\relax
  \def\url#1{\texttt{#1}}\fi
\expandafter\ifx\csname urlprefix\endcsname\relax\def\urlprefix{URL }\fi

\bibitem[{Bertsimas et~al.(2016)Bertsimas, King, and Mazumder}]{bertsimas2016}
Bertsimas, D., King, A., Mazumder, R., 2016. Best subset selection via a modern
  optimization lens. The Annals of Statistics 44~(2), 813--852.

\bibitem[{Bickel et~al.(2009)Bickel, Ritov, and Tsybakov}]{Bickel2009}
Bickel, P.~J., Ritov, Y., Tsybakov, A.~B., 2009. Simultaneous analysis of lasso
  and dantzig selector. The Annals of Statistics 37~(4), 1705--1732.

\bibitem[{Boland et~al.(2015{\natexlab{a}})Boland, Charkhgard, and
  Savelsbergh}]{Boland2015BBM}
Boland, N., Charkhgard, H., Savelsbergh, M., 2015{\natexlab{a}}. A criterion
  space search algorithm for biobjective integer programming: The balanced box
  method. INFORMS Journal on Computing 27~(4), 735--754.

\bibitem[{Boland et~al.(2015{\natexlab{b}})Boland, Charkhgard, and
  Savelsbergh}]{Boland2015TSM}
Boland, N., Charkhgard, H., Savelsbergh, M., 2015{\natexlab{b}}. A criterion
  space search algorithm for biobjective mixed integer programming: The
  triangle splitting method. INFORMS Journal on Computing 27~(4), 597--618.

\bibitem[{Cand{\'e}s and Plan(2009)}]{Candes2009}
Cand{\'e}s, E.~J., Plan, Y., 2009. Near-ideal model selection by $l_1$
  minimization. The Annals of Statistics 37~(5A), 2145--2177.

\bibitem[{Chankong and Haimes(1983)}]{Chankong1983}
Chankong, V., Haimes, Y.~Y., 1983. Multiobjective {D}ecision {M}aking: {T}heory
  and {M}ethodology. Elsevier Science, New York.

\bibitem[{Chen et~al.(1998)Chen, Donoho, and Saunders}]{Chen1998}
Chen, S.~S., Donoho, D.~L., Saunders, M.~A., 1998. Atomic decomposition by
  basis pursuit. SIAM Journal on Scientific Computing 20~(1), 33--61.

\bibitem[{Ciuperca(2014)}]{Ciuperca2014}
Ciuperca, G., May 2014. Model selection by lasso methods in a change-point
  model. Statistical Papers 55~(2), 349--374.

\bibitem[{Dielman(2005)}]{Dielman2005}
Dielman, T.~E., 2005. Least absolute value regression: recent contributions.
  Journal of Statistical Computation and Simulation 75~(4), 263--286.

\bibitem[{Ghosh and Chakraborty(2014)}]{Ghosh2014}
Ghosh, D., Chakraborty, D., 2014. A new pareto set generating method for
  multi-criteria optimization problems. Operations Research Letters 42~(8), 514
  -- 521.

\bibitem[{Hamacher et~al.(2007)Hamacher, Pedersen, and Ruzika}]{Hamacher2007}
Hamacher, W.~H., Pedersen, C.~R., Ruzika, S., 2007. Finding representative
  systems for discrete bicriterion optimization problems. Operations Research
  Letters 35, 336--344.

\bibitem[{Meinshausen and Bühlmann(2006)}]{Meinshausen2006}
Meinshausen, N., Bühlmann, P., 2006. High-dimensional graphs and variable
  selection with the lasso. The Annals of Statistics 34~(3), 1436--1462.

\bibitem[{Miller(2002)}]{Miller2002}
Miller, A., 2002. Subset Selection in Regression, 2nd Edition. Chapman \&
  Hall/CRC Monographs on Statistics \& Applied Probability.

\bibitem[{Ren and Zhang(2010)}]{Ren2010}
Ren, Y., Zhang, X., 2010. Subset selection for vector autoregressive processes
  via adaptive lasso. Statistics \& Probability Letters 80~(23-24), 1705 --
  1712.

\bibitem[{Sayin(1996)}]{SAYIN199687}
Sayin, S., 1996. An algorithm based on facial decomposition for finding the
  efficient set in multiple objective linear programming. Operations Research
  Letters 19~(2), 87 -- 94.

\bibitem[{Schwertman et~al.(1990)Schwertman, Gilks, and Cameron}]{Neil1990}
Schwertman, N.~C., Gilks, A.~J., Cameron, J., 1990. A simple noncalculus proof
  that the median minimizes the sum of the absolute deviations. The American
  Statistician 44~(1), 38--39.

\bibitem[{Stidsen et~al.(2014)Stidsen, Andersen, and Dammann}]{SAD2014}
Stidsen, T., Andersen, K.~A., Dammann, B., 2014. A branch and bound algorithm
  for a class of biobjective mixed integer programs. Management Science 60~(4),
  1009--1032.

\bibitem[{Tibshirani(1996)}]{Tibshirani94}
Tibshirani, R., 1996. Regression shrinkage and selection via the lasso. Journal
  of the Royal Statistical Society, Series B 58, 267--288.

\bibitem[{Wang and Tian(2017)}]{Wang2017}
Wang, M., Tian, G.-L., Feb 2017. Adaptive group lasso for high-dimensional
  generalized linear models. Statistical Papers.

\bibitem[{Zhang and Huang(2008)}]{zhang2008}
Zhang, C.-H., Huang, J., 2008. The sparsity and bias of the lasso selection in
  high-dimensional linear regression. The Annals of Statistics 36~(4),
  1567--1594.

\end{thebibliography}
\end{document}